\documentclass[lettersize,journal]{IEEEtran}
\usepackage{caption}
\usepackage{amsthm}
\usepackage{amsmath,amssymb,amsfonts}
\usepackage{esint}
\usepackage{bm}
\usepackage{upgreek}

\usepackage[linesnumbered,ruled,vlined]{algorithm2e}
\usepackage{array}
\usepackage{textcomp}
\usepackage{subfig}
\usepackage{stfloats}
\usepackage{url}
\usepackage{verbatim}
\usepackage{graphicx}
\usepackage{makecell}
\usepackage{multirow}
\usepackage{xcolor}
\usepackage{cite}
\usepackage{booktabs}
\usepackage{diagbox}
\usepackage{booktabs}
\usepackage{siunitx}
\newtheorem{proposition}{Proposition}
\newtheorem{definition}{Definition}

\theoremstyle{remark}
\allowdisplaybreaks

\newcommand{\figsubref}[2]{Fig.~\ref{#1}\subref{#2}}

\usepackage{etoolbox}
\makeatletter
\newcommand{\patchwarning}[1]{\typeout{*** PATCH WARNING: #1 ***}}

\patchcmd{\@maketitle}
  {\vskip 1.0em\par}
  {\vskip 0.1em\par}
  {}{\patchwarning{Pattern '\string\vskip 1.0em\string\par' not found in \string\@maketitle}}

\patchcmd{\@maketitle}
  {\par \addvspace {0.5\baselineskip }}
  {\par \addvspace {-1.6\baselineskip }}
  {}{\patchwarning{Pattern '\string\par \string\addvspace\space\{0.5\string\baselineskip\space\}' not found in \string\@maketitle}}
\makeatother
\begin{document}

\title{Learning-Based Beamforming for Energy Efficiency of Continuous Aperture Array Systems}

%\author{IEEE Publication Technology,~\IEEEmembership{Staff,~IEEE,}
        % <-this % stops a space
%\thanks{This paper was produced by the IEEE Publication Technology Group. They are in Piscataway, NJ.}% <-this % stops a space
%\thanks{Manuscript received April 19, 2021; revised August 16, 2021.}}

\author{Shiyong~Chen,~\IEEEmembership{Student Member,~IEEE,} Jia~Guo,~\IEEEmembership{Member,~IEEE,}
        and Shengqian~Han,~\IEEEmembership{Senior Member,~IEEE}%

\thanks{Shiyong Chen is with the School of Electronics and Information Engineering, Beihang University, Beijing 100191, China (email: shiyongchen@buaa.edu.cn).}
\thanks{Jia Guo is with the School of Electronics and Information Engineering, Beihang University, Beijing 100191,China (email: guojia@buaa.edu.cn).}
\thanks{Shengqian Han is with the School of Electronics and Information Engineering, Beihang University, Beijing 100191, China (email: sqhan@buaa.edu.cn).}
%\thanks{Manuscript received April 19, 2025; revised August 16, 2025.}
}

% The paper headers
%\markboth{Journal of \LaTeX\ Class Files,~Vol.~14, No.~8, August~2021}%
%Shell \MakeLowercase{\textit{et al.}}: A Sample Article Using IEEEtran.cls for IEEE Journals}

%\IEEEpubid{0000--0000/00\$00.00~\copyright~2021 IEEE}
% Remember, if you use this, you must call \IEEEpubidadjcol in the second
% column for its text to clear the IEEEpubid mark.

\maketitle
\begin{abstract}
This paper jointly optimizes the base-station (BS) continuous aperture array (CAPA) dimensions and beamforming functions to maximize energy efficiency (EE) of the downlink multiuser multi-CAPA system, where both the BS and the users are equipped with CAPAs. Since the beamforming functions are continuous current distribution over the BS CAPA, the resulting EE maximization problem is a nontrivial functional optimization problem that couples aperture sizing and beamforming design. To address this challenge, we propose a cascaded network architecture consisting of a graph neural network (GNN) and a functional-gradient based implicit neural representation (FGB-INR) to learn the BS CAPA dimensions and beamforming functions, respectively. Both networks exploit the permutation equivariance of the optimal optimization policy, and the update equations of FGB-INR are designed according to the functional-gradient structure of the EE objective. Simulation results show that the proposed method approaches the EE of the numerical method while substantially reducing inference latency. They also demonstrates that the functional-gradient structure in FGB-INR improves EE while reducing sample complexity and training time.
\end{abstract}


\section{Introduction}
Continuous aperture arrays (CAPAs) have emerged as a potential antenna architecture for improving spectral and energy efficiency (EE)~\cite{Continuous_aperture_arrays}. Unlike conventional spatially discrete arrays (SPDAs), CAPA uses a nearly continuous electromagnetic aperture to generate controllable current distributions over a finite physical surface, thereby enabling more finely shaped~beamforming.

CAPA beamforming is defined by a continuous current distribution over the aperture. The continuous nature make the beamforming optimization in CAPA systems a functional optimization problem that cannot be solved directly by conventional convex optimization methods~\cite{Beamforming_Optimization}.

To address this challenge, several methods have been proposed for CAPA beamforming. In~\cite{Wavenumber, On_the_SE, Pattern_Division}, the continuous channel and beamforming functions are expanded over finite Fourier bases, which transforms the functional optimization problem into a finite-dimensional one. Although mathematically tractable, this Fourier-based strategy degrades performance because the expansion introduces approximation errors. To avoid such loss, functional weighted minimum mean-squared error (WMMSE) algorithms have been developed~\cite{Beamforming_Design,Functional_WMMSE}. For a single-user multi-CAPA system, in which a CAPA-equipped base station (BS) serves one CAPA-equipped user, a functional WMMSE algorithm was proposed with closed-form update equations derived via the calculus of variations~\cite{Beamforming_Design}. In~\cite{Functional_WMMSE}, this algorithm was extended to a multiuser multi-CAPA system, where a CAPA-equipped BS serves multiple CAPA-equipped users. The corresponding closed-form updates are derived by transforming the first-order optimality conditions from a parameterized form into a functional form. While these functional algorithms avoid the approximation errors by directly optimizing the continuous beamforming functions, their iterative update equations incur substantial computational cost. To reduce online computational complexity, learning-based methods have also been investigated in CAPA systems~\cite{Multi_User_CAPA,Implicit_Neural_CAPA}. For example, in~\cite{Multi_User_CAPA} the optimal beamforming function was represented as a weighted sum of channel functions, and the resulting coefficients was learned by a graph neural network (GNN). In~\cite{Implicit_Neural_CAPA}, the continuous CAPA beamforming functions were directly parameterized by an implicit neural representation~(INR). 

Most numerical and learn-based methods focus on SE or sum-rate maximization, while EE maximization for CAPA systems remains much less explored. In~\cite{Multi_Group_Multicast}, Dinkelbach's method was applied to optimize the beamforming function for EE maximization. Nevertheless, it uses a simplified power consumption model, and the resulting numerical algorithm is computationally expensive.

This paper studies the EE maximization problem in a multiuser CAPA system, where the BS CAPA dimensions and beamforming functions are jointly optimized. This problem is nontrivial as the beamforming functions are defined over the BS aperture and thus are coupled with the BS CAPA dimensions. To address this challenge, we propose a cascaded network architecture consisting of a GNN and a functional-gradient based implicit neural representation (FGB-INR), to learn BS CAPA dimensions and beamforming functions, respectively. Both networks are designed to exploit the permutation equivariance (PE) of
the optimal policy. Moreover, the update equations of FGB-INR are motivated by the structure of functional-gradient iteration.  Simulation results show that the proposed method approaches the EE of the numerical method with substantially reduced inference latency, and ablation study demonstrates that the functional-gradient structure in FGB-INR improves EE and reduces sample complexity and training time.

\section{System Model and Problem Formulation}
% \begin{figure}%[htbp]
% \centering
% \includegraphics[width=0.4\textwidth]{HMIMO4.pdf}
% \caption{Illustration of CAPA hardware implementations.}  \label{CAPA_implement}
% % \vspace{-0.5cm}
% \end{figure}

\subsection{System Model}
Consider a downlink multiuser multi-CAPA system, where a CAPA-equipped BS serves $K$ CAPA-equipped users. In a Cartesian coordinate system, the rectangular BS aperture is placed on the $xy$-plane and centered at the origin, with side lengths $L_{\mathrm B}^x$ and $L_{\mathrm B}^y$ along the $x$- and $y$-axes, respectively. A point on the BS aperture is denoted by $\mathbf{s}=[s_x, s_y, 0]^{\mathsf{T}}$, and the set of all such points is denoted by $\mathcal{S}_\mathrm{B}$. The $k$-th user's CAPA is centered at $\mathbf{r}_{o}^k$ and has side lengths $L_{\mathrm{U}}^{x}$ and $L_{\mathrm{U}}^{y}$. To ensure polarization alignment, each user CAPA is assumed to be parallel to the BS~CAPA~\cite{Beamforming_Design}.

\subsubsection{Sum Rate Model}
Before presenting the achievable sum rate, we first define the inverse of a continuous kernel.

\begin{definition}\label{Definition:Inverse of continuous kernel}
For a continuous kernel $G(\mathbf r, \mathbf s)$,\footnote{In functional analysis, a two-variable function used in an integral operator is referred to as the kernel of that operator.} a kernel $G^{-1}(\mathbf z, \mathbf r)$ is defined as the inverse of $G(\mathbf r, \mathbf s)$ if it satisfies~\cite{Optimal_Beamforming}
\begin{equation}\label{inverse of function}
 \textstyle\int_{\mathcal S} G^{-1}(\mathbf z,\mathbf r)G(\mathbf r,\mathbf s)\mathrm d\mathbf r = \delta(\mathbf z-\mathbf s),
\end{equation}
for all $\mathbf{r}, \mathbf{s}, \mathbf{z}\in\mathcal{S}$, where $\delta(\cdot)$ is the Dirac delta function.
\end{definition}

With Definition~\ref{Definition:Inverse of continuous kernel}, the achievable sum rate of the considered multiuser multi-CAPA system is given by~\cite{Multi_User_CAPA}

\begin{equation}\label{eq:sum_rate}
\textstyle R_{\mathrm{sum}}=\sum_{k=1}^{K} \log\det\big(\mathbf{I}_{d} +\mathbf{Q}_k\big),
\end{equation}
where
\begin{subequations}\label{sum_rate_model}
\begin{align}
&\textstyle\mathbf{Q}_k=
 \iint_{\mathcal{S}_{\mathrm{U}}}
\mathbf{a}_{k,k}^{\mathsf{H}}(\mathbf{r}_1)\,
\mathrm{J}_{{k}}^{-1}(\mathbf{r}_1,\mathbf{r}_2)\,
\mathbf{a}_{k,k}(\mathbf{r}_2)\,
\mathrm{d}\mathbf{r}_2\,\mathrm{d}\mathbf{r}_1,\label{eq_Q}\\
&\textstyle\mathrm{J}_{{k}}(\mathbf{r}_1,\mathbf{r}_2)= \sum\limits_{i=1,i\neq k}^K \mathbf{a}_{k,i}(\mathbf{r}_1)\,\mathbf{a}_{k,i}^{\mathsf{H}}(\mathbf{r}_2) 
+{\sigma_n^{2}}\delta(\mathbf{r}_1-\mathbf{r}_2),\label{eq_J} \\
&\textstyle \mathbf{a}_{k,i}(\mathbf{r}) 
=  \int_{\mathcal{S}_{\mathrm{B}}}h_k(\mathbf{r},\mathbf{s})\,\mathbf{v}_i(\mathbf{s})\,\mathrm{d}\mathbf{s},\label{eq_a}
\end{align}
\end{subequations} and $\mathcal{S}_{\mathrm{U}}=\bigcup_{k=1}^K{\mathcal{S}_{\mathrm{U}}^k}$.
$\mathbf{v}_{k}(\mathbf{s})\in\mathbb{C}^{1\times d}$ denotes the continuous beamforming function for user $k$ with $d$ data streams, and $\sigma_n^2$ is the noise variance. $h_k(\mathbf{r}, \mathbf{s})$ denotes the continuous channel kernel from point $\mathbf{s}$ on the BS CAPA to point $\mathbf{r}$ on the $k$-th user CAPA. As widely used in~\cite{Beamforming_Design, Multi_User_CAPA}, we adopt a line-of-sight propagation model with a uni-polarized CAPA aligned along the $y$-axis. The channel kernel is modeled as
\begin{equation} \label{Green's function}
\!\!\!\textstyle h_k(\mathbf{r}, \mathbf{s})\!=\!\mathbf{u}^{\mathsf{T}}\frac{-j \eta e^{-j \frac{2\pi}{\lambda} \| \mathbf{r} - \mathbf{s} \|}}{2\lambda \| \mathbf{r} - \mathbf{s} \|} 
\Big(\!\mathbf{I}_3-\frac{(\mathbf{r} - \mathbf{s})(\mathbf{r} - \mathbf{s})^{\mathsf T}}{\| \mathbf{r} - \mathbf{s} \|^2} \!\Big)\mathbf{u},
\end{equation}
where $\mathbf{r}\in\mathcal{S}_{\mathrm{U}}^k$, $\mathbf{s}\in \mathcal{S}_{\mathrm{B}}$, $\mathbf{u} = [0, 1, 0]^{\mathsf{T}}$ denotes the unit polarization vector, $\eta$ is the intrinsic impedance of free space, $\lambda$ is the wavelength, and $\mathbf{I}_3$ denotes the $3\times 3$ identity matrix.

\subsubsection{Power Consumption Model}
The total power consumption of the BS transmitter consists of the RF-chain power, BS CAPA surface power, and power-amplifier power~\cite{Continuous_aperture_arrays}. For a BS with $N_{\mathrm{RF}}$ RF chains, its total power consumption is~\cite{Embracing_Reconfigurable_Antennas}
\begin{equation}\label{eq:sum power}
P_{\mathrm{tot}}
=P_{\mathrm{LO}}
+
N_{\mathrm{RF}}\big(2P_{\mathrm{DAC}}+P_{\mathrm{RF}}\big)
+
P_{\mathrm{CAPA}}
+
\frac{P_{\mathrm{rad}}}{\xi},
\end{equation}
where $P_{\mathrm{CAPA}}$ denotes the power consumed by the BS CAPA surface, $P_{\mathrm{rad}}$ is the radiated power, $\xi$ is the power-amplifier efficiency, and $P_{\mathrm{LO}}$, $P_{\mathrm{DAC}}$, and $P_{\mathrm{RF}}$ denote the power consumed by the local oscillator, digital-to-analog converters, and each RF chain, respectively.

The power consumed by the CAPA surface is determined by the field-programmable gate array (FPGA) control board, drive circuits, and metasurface unit cells~\cite{HMIMO_How_Many}. The FPGA control board consumes an approximately constant power, whereas the drive circuits and metasurface unit cells consume power that scales approximately linearly with the number of integrated elements. Since this number scales linearly with CAPA area, $P_{\mathrm{CAPA}}$ is modeled~as
\begin{equation} \label{eq:CAPA power}
P_{\mathrm{CAPA}}
=
P_{\mathrm{CB}}+\alpha L_{\mathrm{B}}^xL_{\mathrm{B}}^y,
\end{equation}
where $P_{\mathrm{CB}}$ denotes the constant power consumed by the FPGA control board, and $\alpha$ is the CAPA surface power consumption coefficient, i.e., the power consumed per unit aperture area.

To characterize the radiated power, we invoke Poynting's theorem, which states that the total radiated power equals the net outward electromagnetic power flow through a closed surface enclosing the radiator~\cite{How_an_antenna_launches}. To compute the radiated power, we enclose the CAPA with a closed surface $\mathcal{S}_{\mathrm c}$ and define the radiation kernel from the CAPA to this surface as
\begin{equation}\label{eq:radiated power}
\textstyle\mathbf G(\mathbf r,\mathbf s)=\frac{-j \eta e^{-j \frac{2\pi}{\lambda} \| \mathbf{r} - \mathbf{s} \|}}{2\lambda \| \mathbf{r} - \mathbf{s} \|}\left(\mathbf I_3-
\frac{(\mathbf r-\mathbf s)(\mathbf r-\mathbf s)^{\mathsf T}}{\|\mathbf r-\mathbf s\|^2}\right)\mathbf u,\,\, \mathbf r\in\mathcal S_{\mathrm c}.
\end{equation}
With~\eqref{eq:radiated power}, the electric field on $\mathcal{S}_{\mathrm c}$ is
\begin{equation}\label{eq:e_c}
\mathbf a_{\mathrm c}(\mathbf r)
=
\sum_{k=1}^{K}
\int_{\mathcal S_{\mathrm B}}
\mathbf G(\mathbf r,\mathbf s)\,
\mathbf v_k(\mathbf s)\mathbf x_k^{\mathsf T}
\,\mathrm d\mathbf s,\,\,\mathbf r\in\mathcal S_{\mathrm c},\,\mathbf s\in\mathcal S_{\mathrm B},
\end{equation}
where $\mathbf{x}_k=[x_{1}^k, \ldots, x_{d}^k]\in\mathbb{C}^{1\times d}$ contains the data symbols for user $k$, with $\mathbb{E}\{|x_i^k|^2\} = 1$. When $\mathcal{S}_{\mathrm c}$ lies in the far-field region, the radiated field is transverse to the propagation direction, and the radiated power is given by~\cite{Pattern_Division}
\begin{equation}\label{eq:radiated_power}
P_{\mathrm{rad}}
=
\mathbb E_{x_i^k}\!\left[
\int_{\mathcal S_{\mathrm c}}
\frac{\|\mathbf a_{\mathrm t}(\mathbf r)\|^2}{2\eta}
\,\mathrm d\mathbf{r}
\right],\,\,\mathbf r\in\mathcal S_{\mathrm c}
\end{equation}
where $\mathbf a_{\mathrm t}(\mathbf r)$ denotes the tangential component of $\mathbf a_{\mathrm c}(\mathbf r)$ on $\mathcal S_{\mathrm c}$, obtained by projecting $\mathbf a_{\mathrm c}(\mathbf r)$ onto the tangent plane of $\mathcal S_{\mathrm c}$.

\subsection{Problem Formulation}
Based on the above models, the joint optimization problem for EE maximization is formulated as
\begin{subequations}\label{P1:maximization of EE}
\begin{align}
\max_{L_{\mathrm{B}}^{x},\,L_{\mathrm{B}}^{y},\,\{\mathbf{v}_k(\mathbf{s})\}_{k=1}^{K}}\quad
& \frac{R_{\mathrm{sum}}}{P_{\mathrm{tot}}}
\label{P1:EE}\\
\mathrm{s.t.}\quad
& L_{\mathrm{B}}^{\min} \leq L_{\mathrm{B}}^{x} \leq L_{\mathrm{B}}^{\max}, \label{P1:Lmax_x}\\
& L_{\mathrm{B}}^{\min} \leq L_{\mathrm{B}}^{y} \leq L_{\mathrm{B}}^{\max}, \label{P1:Lamx_y}\\
&\sum\limits_{k=1}^K\| \mathbf{v}_k(\mathbf{s}) \|^2 \leq \mathrm{I}_{\max},\, \forall \mathbf{s} \in{\mathcal{S}_\mathrm{B}}, \label{P1:PeakConstraint}\\
& \eqref{eq:sum_rate},\eqref{sum_rate_model},\eqref{eq:sum power},\eqref{eq:CAPA power},\eqref{eq:e_c},\eqref{eq:radiated_power},\nonumber
\end{align}
\end{subequations}
where $L_{\mathrm{B}}^{\min}$ and $L_{\mathrm{B}}^{\max}$ are the minimum and maximum side lengths of the BS CAPA, respectively, and $\mathrm{I}_{\max}$ denotes the peak transmit current budget~\cite{Implicit_Neural_CAPA}.

\section{Joint Optimization Method}
In this section, we design a cascaded network architecture, consisting of a GNN followed by FGB-INR, to learn the optimal policy of problem~\eqref{P1:maximization of EE}. We first establish the PE property of this policy, and then exploit  this property to design the GNN and FGB-INR, which are used to learn the CAPA dimension and parameterize the continuous beamforming function, respectively.

\subsection{Permutation Property of the Optimal Policy}
Since the beamforming functions are defined on \(\mathcal{S}_{\mathrm B}\), their domain is coupled with the BS CAPA dimensions, i.e., \(L_{\mathrm B}^x\) and \(L_{\mathrm B}^y\). To decouple the coordinate domain
from the aperture dimensions, we introduce a normalized reference set \(\tilde{\mathcal S}_{\mathrm B}\), where \(\tilde{\mathbf s}=[\tilde s_x,\tilde s_y,0]^{\mathsf T}\in\tilde{\mathcal S}_{\mathrm B}\) with \(-\frac{1}{2}\leq \tilde s_x,\tilde s_y\leq \frac{1}{2}\). The corresponding physical coordinate \(\mathbf{s}\in\mathcal{S}_\mathrm{B}\) is obtained as 
\begin{equation}\label{eq:BS_mapping}
\mathbf s
=
\mathbf T(\tilde{\mathbf s})
=
[L_{\mathrm B}^x\tilde s_x,\,
L_{\mathrm B}^y\tilde s_y,\,
0]^{\mathsf T},
\qquad
\tilde{\mathbf s}\in\tilde{\mathcal S}_{\mathrm B}.
\end{equation}

With this mapping, beamforming functions associated with different aperture sizes can be described on the same normalized domain. Let $\mathbf{R}_o=[\mathbf{r}_o^{1,\mathsf T},\cdots,\mathbf{r}_o^{K,\mathsf T}]^{\mathsf T}\in\mathbb{R}^{K\times 3}$ denote the user geometry, and $\mathbf{V}(\mathbf{s})=[\mathbf{v}_{1}^{\mathsf T}(\mathbf{s}),\ldots,\mathbf{v}_{K}^{\mathsf T}(\mathbf{s})]^{\mathsf T}\in\mathbb{C}^{K\times d}$ denote the beamforming function matrix. Since the optimal $L_{\mathrm{B}}^{x}$, $L_{\mathrm{B}}^{y}$ and $\mathbf{V}(\mathbf{s})$ are determined by the user geometry $\mathbf{R}_o$~\cite{Implicit_Neural_CAPA}, composing $\mathbf{V}(\mathbf{s})$ with~\eqref{eq:BS_mapping} yields the optimal policy as
\begin{equation}\label{optimal_policy}
\big(L_{\mathrm{B}}^{x,\star},L_{\mathrm{B}}^{y,\star},\mathbf{V}^{\star}(\mathbf T(\tilde{\mathbf s}))\big)
=
\mathcal{F}(\tilde{\mathbf{s}},\mathbf{R}_o),
\qquad \tilde{\mathbf{s}}\in \tilde{\mathcal S}_{\mathrm B},
\end{equation}
where $L_{\mathrm{B}}^{x,\star}$ and $L_{\mathrm{B}}^{y,\star}$ are the optimal BS CAPA dimensions, and $\mathbf{V}^{\star}(\mathbf T(\tilde{\mathbf s}))$ is the optimal beamforming function matrix for input $(\tilde{\mathbf{s}},\mathbf{R}_o)$.

The following proposition states that the optimal policy~\eqref{optimal_policy} satisfies the joint permutation invariance (PI) and PE (PIPE) property with respect to the user dimension.

\begin{proposition}\label{proposition:PIPE}
When the input of the policy is permuted as $\boldsymbol{\Pi}^{\mathsf{T}}\mathbf{R}_o$, the BS CAPA dimensions $L_{\mathrm{B}}^{x,\star}$ and $L_{\mathrm{B}}^{y,\star}$ together with $\boldsymbol{\Pi}^{\mathsf{T}}\mathbf{V}^{\star}(\mathbf{s})$ remain optimal, i.e.

\begin{equation}\label{1D-PE}
\textstyle\big(L_{\mathrm{B}}^{x,\star},L_{\mathrm{B}}^{y,\star},\boldsymbol{\Pi}^{\mathsf{T}}\mathbf{V}^{\star}(\mathbf{T}(\tilde{\mathbf s}))\big)=\mathcal{F}\big(\tilde{\mathbf s}, \boldsymbol{\Pi}^{\mathsf{T}}\mathbf{R}_{o}\big),
\end{equation}
where $\boldsymbol{\Pi}\!\in\!\mathbb{R}^{K\times K}$ is a permutation matrix on the user~indices.
\end{proposition}
\begin{IEEEproof}
The proof follows from the invariance of the objective and constraints under the corresponding index permutations and is omitted for brevity.
\end{IEEEproof}

Proposition~\ref{proposition:PIPE} shows that the BS CAPA dimensions are PI, whereas the beamforming functions are PE with respect to the user permutations.

\subsection{Design of the Cascaded Network}
The overall cascaded network architecture is illustrated in Fig.~\ref{Cascaded_network_architecture}, where the first module learns $(L_{\mathrm{B}}^{x},L_{\mathrm{B}}^{y})$, and FGB-INR parameterizes $\mathbf{V}(\cdot)$, as detailed in in Sec.~\ref{Design of FGB-INR}.

\begin{figure}%[htbp]
\centering
\includegraphics[width=0.35\textwidth]{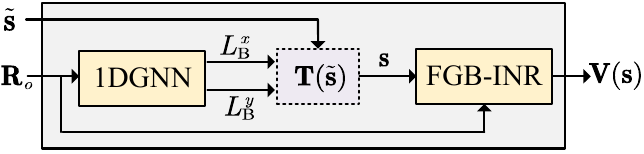}
\caption{Cascaded network architecture.}  \label{Cascaded_network_architecture}
% \vspace{-0.5cm}
\end{figure}

Since the BS CAPA dimensions are determined by the user geometry, the first
module learns the dimension policy
\begin{equation}\label{eq:dimension_policy}
\textstyle\big(L_{\mathrm{B}}^{x},L_{\mathrm{B}}^{y}\big)
=\mathcal{F}_1\big( \boldsymbol{\Pi}^{\mathsf{T}}\mathbf{R}_{o}\big).
\end{equation}
This policy shows that the outputs are permutation invariant with respect to the permutation of the input. To
preserve this property, the first module uses a GNN followed by mean pooling.
The GNN operates on a fully connected user graph with \(K\) vertices and
pairwise edges. For vertex \(k\), the input feature is the user location
\(\mathbf{r}_o^{k}\), and the corresponding output is
\(\mathbf{L}_k\in\mathbb{R}^{2\times 1}\). The BS CAPA dimensions are then
obtained, by averaging the outputs over users indices as
\begin{equation}\label{eq:dimension_pooling}
L_{\mathrm{B}}^{x}=\frac{1}{K}\sum\nolimits_{k=1}^K{\mathbf{L}_{k}}[0],
\quad
L_{\mathrm{B}}^{y}=\frac{1}{K}\sum\nolimits_{k=1}^K{\mathbf{L}_{k}}[1].
\end{equation}

With the learned \((L_{\mathrm{B}}^{x},L_{\mathrm{B}}^{y})\),
the normalized coordinate \(\tilde{\mathbf{s}}\) is mapped to the physical coordinate \(\mathbf{s}\in\mathcal S_{\mathrm B}\) through~\eqref{eq:BS_mapping}. Using $(\mathbf s, \mathbf{R}_{o})$ as input, FGB-INR learns the beamforming policy
\begin{equation}\label{Beamforming policy}
\textstyle\boldsymbol{\Pi}^{\mathsf{T}}\mathbf{V}(\mathbf s)=\mathcal{F}_2\big(\mathbf s, \boldsymbol{\Pi}^{\mathsf{T}}\mathbf{R}_{o}\big), \quad \mathbf{s}\in\mathcal S_{\mathrm B}.
\end{equation}
This beamforming policy maps spatial coordinates to function values and can be parameterized by an INR as~\cite{Implicit_Neural_CAPA}
\begin{equation}\label{eq:BeamINR}
\textstyle\mathbf{V}(\mathbf{s}) = \mathcal{P}_{\theta}(\mathbf{s}, \mathbf{R}_o),\quad\mathbf{s} \in \mathcal{S}_{\mathrm{B}},
\end{equation}
where $\mathcal{P}_{\theta}(\cdot)$ denotes the INR parameterized by $\theta$.

To preserve the PE property in~\eqref{Beamforming policy}, we implement $\mathcal{P}_{\theta}(\cdot)$ using a GNN. This GNN is defined on the same full connected user graph but uses different node features and outputs. Specifically, for vertex $k$, the feature is $(\mathbf{r}_o^k,\mathbf{s})$, and the output is $\mathbf{v}_k(\mathbf{s})$. No features or actions are assigned to the edges.

A conventional GNN updates vertex representation by aggregating information from neighboring vertices and combining it with the vertex's own representation. For user vertex $k$, the hidden representation in the $(l+1)$-th layer, ${\mathbf{d}}_{k}^{(l+1)}(\mathbf{s}) = [d_{k,1}^{(l+1)}(\mathbf{s}), \dots, d_{k,C_{l+1}}^{(l+1)}(\mathbf{s})]^\mathsf{T}$, is updated as
\begin{equation}\label{eq:gnn_update}
\textstyle\mathbf{d}_{k}^{(l+1)}(\mathbf{s})\!=\!\sigma \Big(\mathbf{S}^{(l)}\mathbf{d}_{k}^{(l)}(\mathbf{s})+{{\mathbf{W}^{(l)}\sum\nolimits_{i=1,i\ne k}^K{\mathbf{d}_{i}^{(l)}}(\mathbf{s})}} \Big),
\end{equation}	
where $C_l$ is the representation dimension in the $l$-th layer, 
${\bf S}^{(l)}, {\bf W}^{(l)}\in\mathbb{C}^{C_{l+1}\times C_{l}}$ are trainable parameters, and $\sigma(\cdot)$ is an element-wise activation function. The initial and final hidden representations, ${\mathbf{d}}_{k}^{(0)}(\mathbf{s})$ and ${\mathbf{d}}_{k}^{(L)}(\mathbf{s})$, correspond to the vertex features and outputs, respectively, where $L$ is the number of the GNN layers.

\subsection{Design of FGB-INR}\label{Design of FGB-INR}
As analyzed in~\cite{Gradient_GNN}, the gradient-iteration equation for SPDA beamforming can guide the design of GNN update equation. We therefore derive the functional gradient of the EE objective via variational calculus and use the resulting iteration to design the update equations of FGB-INR. For the $k$-th user, the functional-gradient iteration with respect to $\mathbf{v}_k(\mathbf{s})$ is written as
\begin{equation}\label{eq:functional iteration equation}
\textstyle\mathbf{v}^{(l+1)}_k(\mathbf{s})=\mathbf{v}^{(l)}_k(\mathbf{s})+\eta\frac{\delta \mathcal{J}}{\delta \mathbf{v}^{(l)}_k}(\mathbf{s}),\quad \mathbf{s} \in \mathcal{S}_{\mathrm{B}},
\end{equation}
where $\mathbf{v}^{(l)}_k(\mathbf{s})$ denotes the beamforming function at the $l$-th iteration, $\eta>0$ is the step size, $\mathcal{J}$ is the objective in~\eqref{P1:EE}, and $\frac{\delta \mathcal{J}}{\delta \mathbf{v}^{(l)}_k}(\mathbf{s})$ denotes the functional gradient evaluated at $\mathbf{v}^{(l)}_k(\mathbf{s})$. This functional gradient is derived as
\begin{equation}\label{eq:functional gradient equation}
\begin{split}\raisetag{3.5ex}
\textstyle\frac{\delta \mathcal{J}}{\delta \mathbf{v}^{(l)}_k}(\mathbf{s})=&\textstyle\int_{\mathcal{S} _{\mathrm{U}}}{\!h_{k}^{\mathsf{H}}(\mathbf{r},\mathbf{s})\bm{\upalpha}^{(l)} _k\left( \mathbf{r} \right) \mathrm{d}\mathbf{r}}+\int_{\mathcal{S}_{\mathrm{c}}}{\!h_{c}^{}(\mathbf{r},\mathbf{s})\bm{\upvarphi}^{(l)} _k\left( \mathbf{r} \right)}\mathrm{d}\mathbf{r}\\
    &\textstyle+\sum_{i=1,i\ne k}^K{\int_{\mathcal{S} _{\mathrm{U}}}{\!h_{i}^{\mathsf{H}}(\mathbf{r},\mathbf{s})\bm{\upbeta}^{(l)} _{k,i}\left( \mathbf{r} \right) \mathrm{d}\mathbf{r}}},
\end{split}
\end{equation}
where the functions $\bm{\upalpha}^{(l)}_k\!(\mathbf r)$, 
$\bm{\upbeta}^{(l)}_{k,i}(\mathbf r)$, and 
$\bm{\upvarphi}^{(l)}_k\!(\mathbf r)$ are defined~as
\vspace{-0.4cm}
\begin{subequations}\label{functions}
\begin{align}
&\textstyle\bm{\upalpha}^{(l)}_k(\mathbf r)
=
\frac{1}{P_{\mathrm{tot}}}
\int_{\mathcal S_{\mathrm U}}
\mathrm{J}_{k}^{-1}(\mathbf r,\mathbf r_1)
\mathbf a_{k,k}(\mathbf r_1)
\,\mathrm d\mathbf r_1\,
\mathbf O_{k}^{-1,\mathsf H},
\label{alpha}\\
&\textstyle\bm{\upbeta}^{(l)}_{i,k}(\mathbf r)
=
\frac{-1}{P_{\mathrm{tot}}}
\int_{\mathcal S_{\mathrm U}}
\mathrm{J}_{i}^{-1}(\mathbf r,\mathbf r_1)
\mathbf a_{i,i}(\mathbf r_1)
\,\mathrm d\mathbf r_1\,
\mathbf O_{k}^{-1,\mathsf H}\mathbf P_{i,k},
\label{beta}\\
&\textstyle\bm{\upvarphi}^{(l)}_k(\mathbf r)
=
\frac{-R_{\mathrm{sum}}}{P_{\mathrm{tot}}^{2}}
\int_{\mathcal S_{\mathrm B}}
h_{c}(\mathbf r,\mathbf s)
\mathbf v_k(\mathbf s)
\,\mathrm d\mathbf s\label{varphi},
\end{align}
\end{subequations}
with $\mathbf{P}_{i,k}\!\!=\!\!\iint_{\mathcal{S}_{\mathrm{U}}}\!\!\mathbf{a}_{i,i}^{\mathsf H}(\mathbf{r})\mathrm{J}_{i}^{-1}(\mathbf{r},\mathbf{r}_1)\mathbf{a}_{i,k}(\mathbf{r}_1)\mathrm{d}\mathbf{r}_1\mathrm{d}\mathbf{r}$ and $\mathbf{O}_k=\mathbf{Q}_k+\mathbf{I}_d$.
Substituting~\eqref{eq:functional gradient equation} into~\eqref{eq:functional iteration equation} yields
\begin{equation}
\begin{split}\raisetag{1.35cm}\label{eq:functional iteration equation1}
&\textstyle\mathbf{v}^{(l+1)}_k(\mathbf{s})=\mathbf{v}^{(l)}_k(\mathbf{s})+\eta\int_{\mathcal{S} _{\mathrm{U}}}{\!h_{k}^{\mathsf{H}}(\mathbf{r},\mathbf{s})\bm{\upalpha}^{(l)} _k\left( \mathbf{r} \right) \mathrm{d}\mathbf{r}}+\\
&\textstyle\eta\int_{\mathcal{S}_{\mathrm{c}}}{\!h_{c}^{}(\mathbf{r},\mathbf{s})\bm{\upvarphi}^{(l)} _k\left( \mathbf{r} \right)}\mathrm{d}\mathbf{r}+\!\!\!\!\sum\limits_{i=1,i\ne k}^K\!\!\!\eta{\int_{\mathcal{S} _{\mathrm{U}}}{\!h_{i}^{\mathsf{H}}(\mathbf{r},\mathbf{s})\bm{\upbeta}^{(l)} _{i,k}\left( \mathbf{r} \right) \mathrm{d}\mathbf{r}}}.
\end{split}
\end{equation}

Comparing~\eqref{eq:functional iteration equation1} with~\eqref{eq:gnn_update}, we observe two structural similarities. Both updates include an aggregation term that collects information from other users, and both combine the aggregated information with user-specific terms to update the current user. They also differ in two aspects: the aggregation term in~\eqref{eq:functional iteration equation1} has a different form from that in~\eqref{eq:gnn_update}, and~\eqref{eq:functional iteration equation1} contains three combination terms, whereas~\eqref{eq:gnn_update} contains only one.

This observation motivates a new aggregation and combination design. Replacing $\mathbf{v}^{(l)}_k(\mathbf{s})$ in~\eqref{eq:functional iteration equation1} with the hidden representation $\mathbf{d}^{(l)}_k(\mathbf{s})$ and introducing trainable parameters and an activation function~yields
\begin{equation}
\begin{split}\raisetag{0.5cm}\label{update equation of FGB-INR}
\mathbf{d}^{(l+1)}_k(\mathbf{s})&\textstyle=\sigma\Big(\mathbf{S}_1^{(l)}\mathbf{d}^{(l)}_k(\mathbf{s})+\mathbf{S}_2^{(l)}\int_{\mathcal{S} _{\mathrm{U}}}{\!h_{k}^{\mathsf{H}}(\mathbf{r},\mathbf{s})\mathbf{b}^{(l)}_k\left( \mathbf{r} \right) \mathrm{d}\mathbf{r}}+\\
&\textstyle+\mathbf{W}_1^{(l)}\sum\nolimits_{i=1,i\ne k}^K\!{\int_{\mathcal{S} _{\mathrm{U}}}{\!h_{i}^{\mathsf{H}}(\mathbf{r},\mathbf{s})\mathbf{q}^{(l)}_{k,i}\left( \mathbf{r} \right) \mathrm{d}\mathbf{r}}}\Big),
\end{split}
\end{equation}
where ${\bf S}_1^{(l)}, {\bf S}_2^{(l)}, {\bf W}_1^{(l)}\in\mathbb{C}^{C_{l+1}\times C_l}$ are trainable matrices, and $\mathbf{b}^{(l)}_k(\mathbf r)\in\mathbb{C}^{C_l\times 1}$ and $\mathbf{q}^{(l)}_{k,i}(\mathbf r)\in\mathbb{C}^{C_l\times 1}$ are functions obtained by replacing $\mathbf{v}^{(l)}_k(\mathbf s)$ in~\eqref{alpha} and~\eqref{beta} with $\mathbf{d}^{(l)}_k(\mathbf s)$. The step size $\eta$ is absorbed into the trainable matrices $\mathbf{S}_2^{(l)}$ and $\mathbf{W}_1^{(l)}$. Since the beamforming function is independent of the specific choice of the enclosing surface $\mathcal{S}_{\mathrm{c}}$, the counterpart of $\bm{\upvarphi}^{(l)} _k\left( \mathbf{r} \right)$ in~\eqref{eq:functional iteration equation1} is omitted.

As~\eqref{alpha} and~\eqref{beta} involve kernel and matrix inverses, directly evaluating $\mathbf{b}^{(l)}_{k}(\mathbf{r})$ and $\mathbf{q}^{(l)}_{k,i}(\mathbf{r})$ in~\eqref{update equation of FGB-INR} is computationally expensive. We therefore employ two INRs, $\mathcal{B}_{\theta}(\cdot)$ and $\mathcal{Q}_{\theta}(\cdot)$, to parameterize $\mathbf{b}^{(l)}_{k}(\cdot)$ and $\mathbf{q}^{(l)}_{k,i}(\cdot)$, respectively. 
We first describe the design of $\mathcal{B}_{\theta}(\cdot)$ and the same principle applies to $\mathcal{Q}_{\theta}(\cdot)$. 
Substituting~\eqref{eq_J} into~\eqref{alpha}, we find that $\bm{\upalpha}^{(l)}_k(\mathbf r)$ is a function of $\mathbf{A}^{(l)}_{k}(\mathbf r)=[\mathbf{a}^{(l),\mathsf{T}}_{k,1}(\mathbf r),\cdots,\mathbf{a}^{(l),\mathsf{T}}_{k,K}(\mathbf r)]^{\mathsf H}\in\mathbb{C}^{K\times d}$. 
Then replacing \(\mathbf{v}^{(l)}_i(\mathbf s)\) in~\eqref{alpha} and \(\mathbf{A}^{(l)}_{k}(\mathbf r)\) with \(\mathbf{d}^{(l)}_i(\mathbf s)\), \(\bm{\upalpha}^{(l)}_k(\mathbf r)\) is represented by \(\mathbf{b}^{(l)}_{k}(\mathbf r)\), and
\(\mathbf{A}^{(l)}_{k}(\mathbf r)\) is replaced by $\mathbf{E}^{(l)}_{k}(\mathbf r)=
[\mathbf{e}^{(l),\mathsf{T}}_{k,1}(\mathbf r),\cdots,
\mathbf{e}^{(l),\mathsf{T}}_{k,K}(\mathbf r)]^{\mathsf H}
\in\mathbb{C}^{K\times C_l}$
where $\mathbf{e}^{(l)}_{k,i}(\mathbf r)
=
\int_{\mathcal{S}_{\mathrm B}}
h_k(\mathbf r,\mathbf s)\mathbf{d}^{(l)}_i(\mathbf s)
\,\mathrm d\mathbf s$. Thus, \(\mathbf{b}^{(l)}_{k}(\mathbf r)\) is a function
of \(\mathbf{E}^{(l)}_{k}(\mathbf r)\), which can be parameterized by \(\mathcal{B}_{\theta}(\cdot)\) as  
\begin{equation}\label{eq:B_inr}
\mathbf{B}^{(l)}(\mathbf{r})=\sigma\big( \mathbf{W}_b^{(l)}\mathbf{E}^{(l)}_{k}(\mathbf{r})\big),
\end{equation}
where $\mathbf{B}^{(l)}(\mathbf{r})=[\mathbf{b}_1^{(l)}(\mathbf{r}),\cdots,\mathbf{b}_K^{(l)}(\mathbf{r})]\in\mathbb{C}^{K\times C_l}$, and $\mathbf{W}_b^{(l)}\in\mathbb{C}^{K\times K}$ is a trainable matrix. To further preserve the PE property of FGB-INR, $\mathcal{B}_{\theta}(\cdot)$~satisfies
\begin{equation}\label{beta GNN condition}
\boldsymbol{\Pi}^{\mathsf{T}}\mathbf{B}^{(l)}(\mathbf{r})=\sigma\big( \mathbf{W}_b^{(l)}\left(\boldsymbol{\Pi}^{\mathsf{T}}\mathbf{E}^{(l)}_{k}(\mathbf{r})\right)\big),
\end{equation}
which can be enforced by a parameter-sharing scheme for $\mathbf{W}_b^{(l)}$~\cite{Deep_models_of}.
Using the same design principle as $\mathcal{B}_{\theta}(\cdot)$, $\mathcal{Q}_{\theta}(\cdot)$ is designed~as
\begin{equation}\label{eq:Q_inr}
\mathbf{Q}_k^{(l)}(\mathbf{r})
=
\sigma\big(\mathbf{W}_q^{(l)}\mathbf{E}^{(l)}_{k}(\mathbf{r})\big),
\end{equation}
where $\mathbf{Q}^{(l)}_{k}(\mathbf r)=[\mathbf{q}^{(l),\mathsf{T}}_{k,1}(\mathbf r), \cdots, \mathbf{q}^{(l),\mathsf{T}}_{k,K}(\mathbf r)]^{\mathsf H}\in\mathbb{C}^{K\times C_l}$, and $\mathbf{W}_q^{(l)}\in\mathbb{C}^{K\times K}$ is a trainable matrix that follows the same parameter-sharing scheme as $\mathbf{W}_b^{(l)}$.

\section{Simulation Results} \label{Simulation}
This section evaluates the proposed method and compares it with related baselines.

\subsection{Simulation Setup}\label{Simulation_setup}
We consider a BS CAPA services \(K=3\) CAPA-equipped users. The maximum and minimum side lengths of the BS CAPA are \(L_{\mathrm{B}}^{\max}=2\,\mathrm{m}\) and \(L_{\mathrm{B}}^{\min}=0.1\,\mathrm{m}\), respectively, while each user has a square CAPA with \(L_{\mathrm{U}}^{x}=L_{\mathrm{U}}^{y}=0.5\,\mathrm{m}\). The position of user \(k\), \(\mathbf r_{o}^{k}=[r_{x}^{k},\,r_{y}^{k},\,r_{z}^{k}]^{\mathsf T}\), is uniformly distributed in a region with \(r_{x}^{k},r_{y}^{k}\sim\mathcal{U}[-5,5]\,\mathrm{m}\) and \(r_{z}^{k}\sim\mathcal{U}[20,30]\,\mathrm{m}\). The wavelength is set to \(\lambda=0.125\,\mathrm{m}\), corresponding to a carrier frequency of \(2.4\,\mathrm{GHz}\), and the intrinsic impedance is \(\eta=120\pi\,\Omega\). The number of data streams is \(d=2\) and the number of RF chains is \(N_{\mathrm{RF}}=Kd\). The power consumption parameters are \(P_{\mathrm{LO}}=22.5\,\mathrm{mW}\), \(P_{\mathrm{DAC}}=128\,\mathrm{mW}\), \(P_{\mathrm{RF}}=31.6\,\mathrm{mW}\), \(\alpha=20\,\mathrm{W/m^2}\), and \(P_{\mathrm{CB}}=4.8\,\mathrm{W}\)~\cite{Practical_Measurement_Validation,Embracing_Reconfigurable_Antennas}. The power-amplifier efficiency, peak transmit current budget, and noise variance are set to \(\xi=0.27\), \(\mathrm{I}_{\max}=500\,\text{mA}^2\), and \(\sigma_n^{2}=5.6\times10^{-3}\,\mathrm{V}^2\), respectively~\cite{Implicit_Neural_CAPA}.

The GNN in the first module is configured with five hidden layers of dimensions $\{32, 128, 256, 128, 32\}$, and FGB-INR is configured with six hidden layers of dimensions $\{64, 128, 512, 512, 128, 64\}$. Both networks use Tanh activations at each layer. The cascaded network is trained in an unsupervised manner by minimizing the negative weighted sum of two numerical estimates of the objective in~\eqref{P1:EE}. These two estimates are computed using different integral schemes: one with fixed Gauss--Legendre quadrature points and the other with randomized sampling points~\cite{Implicit_Neural_CAPA}. We generate $500,000$ samples for training and use a separate set of $10,000$ testing samples. The model is trained using the Adam optimizer with an initial learning rate of $10^{-3}$ and a batch size of $8$.

\subsection{Learning Performance}
We compare the learning performance of the proposed method with the following baselines:
\begin{itemize}
    \item \textbf{INR-Beam}: A learning-based method from~\cite{Implicit_Neural_CAPA}, which employs a FNN as an INR to parameterize the beamforming functions while leaving the BS CAPA dimensions fixed at $L_{\mathrm{B}}^{x}=L_{\mathrm{B}}^{y}=L_{\mathrm{B}}^{\max}$.

    \item \textbf{Dink-Beam}: A method that uses Dinkelbach's method to optimize the beamforming functions with fixed BS CAPA~dimensions, i.e., $L_{\mathrm{B}}^{x}=L_{\mathrm{B}}^{y}=L_{\mathrm{B}}^{\max}$~\cite{Multi_Group_Multicast}.

    \item \textbf{Nest-Opt.}: A nested joint optimization method, where the outer loop performs a coarse-to-fine search over the BS CAPA dimension with \(L_{\mathrm{B}}^{x}=L_{\mathrm{B}}^{y}\), while the inner loop optimizes the beamforming functions using Dinkelbach's~method.
\end{itemize}

Since the power-consumption coefficient $\alpha$ of the CAPA surface is implementation-dependent and may vary across practical hardware architectures~\cite{Practical_Measurement_Validation, HMIMO_How_Many}, we first evaluate the EE performance under different values of $\alpha$ in~\figsubref{Performance_comparison_under_different_alf_User}{EE_Vs_alf}. As expected, the EE decreases as $\alpha$ increases because a larger $\alpha$ leads to higher CAPA surface power consumption. \figsubref{Performance_comparison_under_different_alf_User}{EE_Vs_User} shows the EE performance versus the number of users. The EE of all considered methods increases with the number of users, indicating that the CAPA system provides sufficient degrees of freedom to support multiuser multiplexing gains. Across both figures, the proposed method and Nest-Opt., which jointly optimize the BS CAPA dimensions and beamforming functions, consistently outperform INR-Beam and Dink-Beam, which only optimize the beamforming functions. In addition, the proposed method achieves comparable EE with Nest-Opt.

\begin{figure}[h]
\centering 
\vspace{-0.2cm} 
\subfloat[EE versus $\alpha$.]{\includegraphics[width=0.46\linewidth]{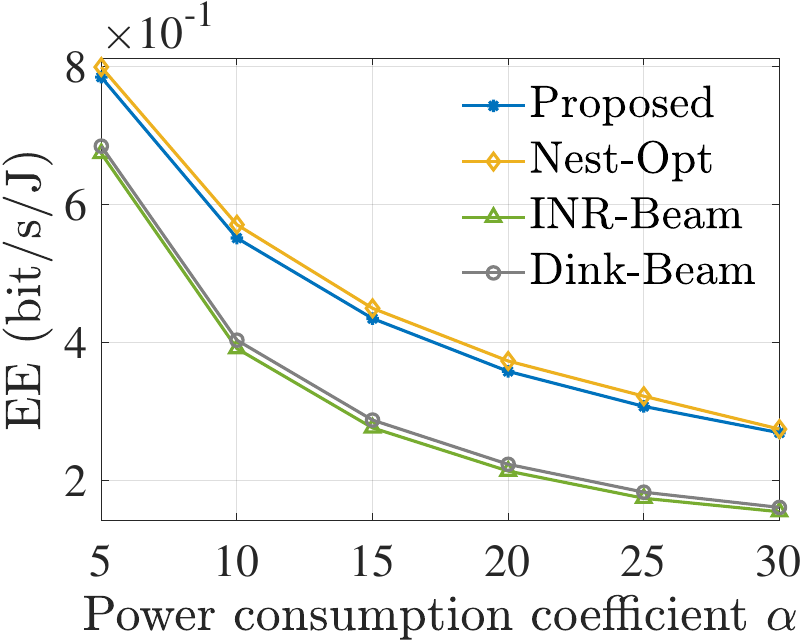} \label{EE_Vs_alf}}
\hfill
\subfloat[EE versus $K$.]{\includegraphics[width=0.45\linewidth]{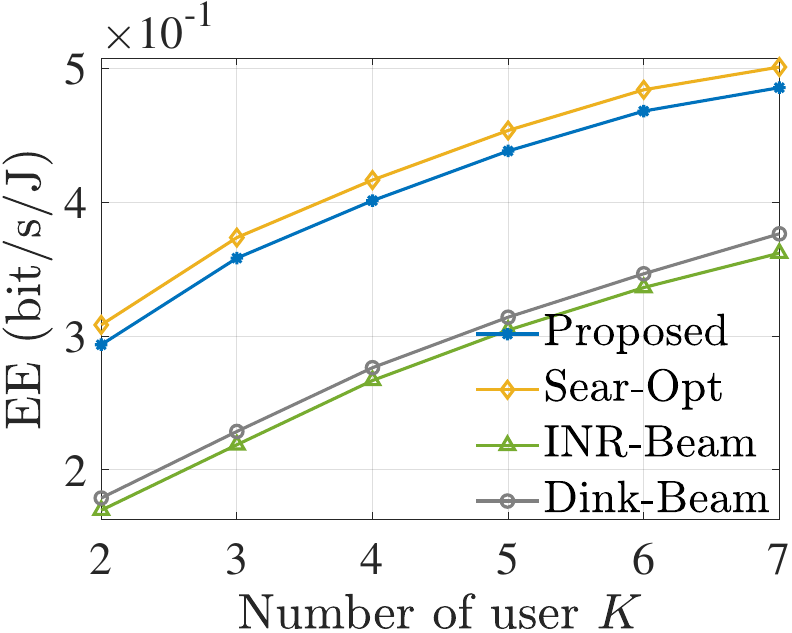}\label{EE_Vs_User}  } 
\caption{Performance comparison under different parameters.} 
\label{Performance_comparison_under_different_alf_User} \vspace{-0.1cm} 
\end{figure}

\subsection{Ablation Study}
In this subsection, we design two ablation variants to validate and quantify the contribution of the two key designs in the proposed method, i.e., the functional-gradient-based update in FGB-INR and the joint PIPE property.

\begin{itemize}
    \item \textbf{VAGNN}: This variant replaces FGB-INR with a conventional GNN using the update equation in~\eqref{eq:gnn_update}. It therefore relies solely on the joint PIPE property and isolates the gain brought by the proposed FGB-INR update.

    \item \textbf{VAFNN}: This variant replaces the networks in both cascaded modules with FNNs. Comparing it with VAGNN isolates the gain brought by the joint PIPE property.
\end{itemize}

Table~\ref{Learning_Performance_Num} shows the learning performance of the considered methods under different numbers of training samples. The performance metric is defined as the ratio of the EE achieved by each learning-based method to Nest-Opt. The proposed method and VAGNN consistently outperform VAFNN across all numbers of training samples, verifying the benefit of exploiting the joint PIPE property. The superior performance of the proposed method over VAGNN demonstrates the the gains provided by the FGB-INR update equations.

\begin{table}[h]
% \vspace{-0.2cm}
\captionsetup{font=small}
\renewcommand{\arraystretch}{1.3}
\centering
\caption{Learning Performance Versus the Number of Training Samples}
\label{Learning_Performance_Num}
\setlength{\tabcolsep}{3pt}
\begin{tabular}{@{}c|*7{c}@{}}
\hline
Number of samples & 5  & 10  & 50  & 100  & 500  & 1000  & 5000 \\
\hline
Proposed (\%) & 69.48 & 81.38 & 90.41 & 90.97 & 92.07 & 92.60 & 95.12 \\
VAGNN (\%)    & 15.87 & 22.45 & 65.27 & 74.33 & 85.08 & 87.83 & 90.62 \\
VAFNN (\%)    & 10.77 & 12.75 & 31.35 & 41.64 & 68.84 & 79.80 & 84.43 \\
\hline
\end{tabular}
% \vspace{-0.2cm}
\end{table}

Table~\ref{Complexity} compares the inference latency and training complexity of the considered methods. For learning-based methods, training complexity is measured by the number of training samples, training time, and model parameters
required to reach \(90\%\) of the EE achieved by Nest-Opt. The results show that all learning-based methods substantially reduce inference latency compared with the numerical baseline. Among the learning-based methods, the proposed method requires the fewest training samples and shortest training time. It exhibits slightly higher inference latency than the two variants and more model parameters than VAGG, due to the tow additional sub-networks in FGB-GNN.

\begin{table}[htbp]	
\centering 
\small
\caption{Inference Latency and Training Complexity}	
\begin{tabular}{c|c|c|c|c}
\hline
\multirow{2}{*}{Name} & \multirow{2}{*}{Inference Latency} & \multicolumn{3}{c}{Training Complexity} \\
\cline{3-5}
                            &                   & Sample   & Time      & Space \\ \hline             
\textbf{Proposed}           & 0.051 s           & 50       & 0.18 h    & 3.19 M    \\ 
\cline{1-5}
\textbf{VAGNN}              & 0.035 s           & 5000      & 0.37 h    & 1.48 M    \\ 
\cline{1-5}\textbf{VAFNN}   & 0.028 s           & 30 K     & 2.01 h    & 7.55 M    \\ 
\cline{1-5}
\textbf{Nest-Opt}           & 3.42 h           & --        & --       & --  \\ 
\hline
\end{tabular}
%\vspace{-0.5em}
\begin{minipage}{0.95\linewidth}
\quad\footnotesize \textit{Note:} K and M represent thousand and million, respectively.
\end{minipage}
\label{Complexity}
\end{table}

\section{Conclusion}
This paper investigated EE maximization in downlink multiuser CAPA systems, where the BS CAPA dimensions and the beamforming functions are jointly optimized. To solve this problem, we developed a cascaded network consisted of a GNN and FGB-INR. Both networks exploit the PE property of the optimal policy, and the update equations of FGB-INR are designed according to the functional gradient structure of the EE objective. Simulation results show that the proposed method achieves comparable EE with the numerical method with substantially reduced inference latency, and ablation study demonstrates that the functional-gradient structure in FGB-INR improves EE and reduces sample complexity and training time. 

\bibliographystyle{IEEEtran}
\bibliography{main}
\end{document}